\documentclass{article}

\usepackage[letterpaper]{geometry}

\def\journal{0}			
\usepackage[normalem]{ulem}

\PassOptionsToPackage{hyphens}{url}
\usepackage[colorlinks=true,citecolor=blue,linkcolor=blue,linktocpage]{hyperref}
\usepackage{makecell}

\usepackage[linesnumbered,ruled,vlined]{algorithm2e}
\SetKwInput{KwInput}{Input}                
\SetKwInput{KwOutput}{Output}
\usepackage{cite}

\usepackage{drawmatrix}
\usepackage{tikz}
\usepackage{xcolor}
\usepackage[export]{adjustbox}
\usetikzlibrary{matrix,positioning}
\usepackage{graphicx}
\usepackage{caption}
\usepackage[labelformat=simple]{subcaption}

\usepackage{epstopdf}
\usepackage{epsfig}
\usepackage{amssymb}
\usepackage{amsmath}
\usepackage{amsthm}
\usepackage{latexsym}
\usepackage{setspace}
\usepackage{bbm}
\usepackage{flushend}
\usepackage{float}
\usepackage{dsfont}
\usepackage{algorithmic,algorithm2e}
\usepackage{xspace}
\usepackage{enumerate}

\usepackage{comment}

\usepackage[title]{appendix}

\newcommand\numberthis{\addtocounter{equation}{1}\tag{\theequation}}

\newcommand{\bPr}[1]{{\mathrm{Pr}}\left(#1\right)}

\newcommand{\bEE}[1]{{\mathbb{E}}\left[#1\right]}

\newcommand{\cL}{{\mathcal L}}

\newcommand{\cN}{{\mathcal N}}
\newcommand{\mN}{{\mathbbm N}}
\newcommand{\mR}{{\mathbbm R}}

\newcommand{\cS}{{\mathcal S}}

\makeatletter
\newtheorem*{rep@theorem}{\rep@title}
\newcommand{\newreptheorem}[2]{%
\newenvironment{rep#1}[1]{%
 \def\rep@title{#2 \ref{##1}}%
 \begin{rep@theorem}}%
 {\end{rep@theorem}}}
\makeatother

\newtheorem{theorem}{Theorem}

\newtheorem{corollary}[theorem]{Corollary}
\newtheorem*{corollary*}{Corollary}
\newtheorem{assumption}{Assumption}

\newtheorem*{assumptions*}{Assumptions}
\newtheorem{lemma}[theorem]{Lemma}
\newtheorem*{lemma*}{Lemma}

\newtheorem{lemma-app}{Lemma}[section]

\newreptheorem{theorem}{Theorem}
\newreptheorem{lemma}{Lemma}

\theoremstyle{remark}

\newtheorem*{remark*}{Remark}
\newtheorem*{remarks*}{Remarks}

\theoremstyle{definition}
\newtheorem{definition}{Definition}

\newcommand{\ed}{\stackrel{{\rm def}}{=}}

\def\undertilde#1{\mathord{\vtop{\ialign{##\crcr
$\hfil\displaystyle{#1}\hfil$\crcr\noalign{\kern1.5pt\nointerlineskip}
$\hfil\tilde{}\hfil$\crcr\noalign{\kern1.5pt}}}}}

\allowdisplaybreaks

\ifnum\journal=1
	\includecomment{onlycnfr}
	\excludecomment{onlyarxiv}
\else
	\includecomment{onlyarxiv}
	\excludecomment{onlycnfr}
\fi

\begin{document}
\title{Phase Transitions for Support Recovery from Gaussian Linear Measurements}
\date{}
\author{{Lekshmi Ramesh} \and
  {Chandra R. Murthy} \and {Himanshu
    Tyagi} }

\maketitle

{\renewcommand{\thefootnote}{}\footnotetext{
\noindent The authors are with the Department of Electrical
Communication Engineering, Indian Institute of Science, Bangalore
560012, India.  Email: \{lekshmi, cmurthy, htyagi\}@iisc.ac.in.} }


{\renewcommand{\thefootnote}{}\footnotetext{
		\noindent \thanks{This work was financially supported by a PhD fellowship from the Ministry of Electronics and Information Technology, Govt. of India, and by research grants from the Aerospace Network Research Consortium and the Center for Networked Intelligence (CNI) at the Indian Institute of Science.}} }
		
{\renewcommand{\thefootnote}{}\footnotetext{
		\noindent \thanks{To appear in ISIT 2021.}} }

\begin{abstract}
  We study the problem of recovering the common $k$-sized support of a set of $n$ samples of dimension $d$, using $m$ noisy linear measurements per sample.
  Most prior work has focused on the case when $m$ exceeds $k$, in which case
      $n$ of the order $(k/m)\log(d/k)$ is both necessary and sufficient.
      Thus, in this regime,
      only the total number of measurements across the samples matter, and there
      is not much benefit in getting more than $k$ measurements per sample.
  In the measurement-constrained regime where we have access to fewer
  than $k$ measurements per sample, we show an upper bound of
  $O((k^{2}/m^{2})\log d)$ on the sample complexity for successful support recovery when $m\ge 2\log d$. Along with the lower bound from our previous work, this
  shows a phase transition for the sample complexity of this problem
  around $k/m=1$. In fact, our proposed algorithm is sample-optimal in both the regimes. It follows that, in the $m\ll k$ regime,
    multiple measurements from the same sample are more valuable
    than measurements from different samples.
  
\end{abstract}
\section{Introduction}\label{sec:introduction}
The problem of support recovery in the single sample setting
considers the following question: given noisy linear measurements
$Y=\Phi x+ W\in\mR^{m}$ of a $k$-sparse vector $x\in\mR^{d}$, can we
recover the locations of its nonzero entries when $m<d$? The set of
indices corresponding to the nonzero entries of $x$ is called the
support of $x$, and is denoted by $\mathtt{supp}(x)$. The measurement matrix
$\Phi\in\mR^{m\times d}$ is a design parameter that is chosen to
enable exact or approximate recovery of $\mathtt{supp}(x)$, and $W\sim\cN(0,\sigma^{2} I)$ is noise. This
problem (also sometimes referred to as model selection or variable
selection) has received a lot of attention in the past decade
\cite{Wainwright_TIT_2009}, \cite{Fletcher_TIT_2009},
\cite{Aeron_TIT_2010}, \cite{Reeves_ISIT_2008},
\cite{Ndaoud_TIT_2020}, with a focus on designing recovery algorithms
and on determining the number of measurements $m$ required for
successful recovery. In particular, it is known that $m=\Theta(k\log(d-k))$ measurements are necessary and sufficient for support recovery
with high probability using a Gaussian measurement matrix \cite{Wainwright_TIT_2009}. 
It is important to note that this tight scaling holds in the low signal to noise ratio (SNR) regime of $x_{\min}/\sigma^{2}=\Theta(1/k)$, where $x_{\min}\ed\min_{i\in\cS}x_{i}$. In other regimes of SNR, either the $\log$ dependence changes or the upper and lower bounds are known to differ by a factor of $(\log(1+k x_{\min}^{2}/\sigma^{2}))^{-1}$; see \cite{Ndaoud_TIT_2020} for a detailed discussion.

Parallel to the results in the single sample setting, there has been work on the natural extension of this problem to the multiple sample setting, which is the focus of this work.
In this setting, there are multiple samples $x_{1},\ldots,x_{n}$, all sharing a common unknown support $\cS$ of cardinality $k$. For each sample $x_{i}$, we observe measurements $Y_{i}=\Phi_{i}x_{i}+W_{i}$, and the goal is to recover $\cS$.
We can ask the question of how the number of measurements per sample $m$ and the number of samples $n$ can be traded-off for each other, and whether it is useful to take more samples or more measurements per sample.

While there have been several works in the multiple sample setting  \cite{Wipf_TSP_2007}, \cite{Tang_TIT_2010}, \cite{Foucart_SAMPTA_2011}, \cite{Jin_TIT_2013}, \cite{Scarlett_TIT_2017}, \cite{Park_TIT_2017}, they focus on the regime where one has access to roughly $m\ge k$ measurements per sample. In particular, omitting the dependence on SNR, \cite{Park_TIT_2017} shows that $mn=\Theta(k\log(d/k))$ is necessary and sufficient assuming $m=\Omega(k)$ and $k=o(d)$. While the sufficient condition in \cite{Park_TIT_2017} is obtained via analysis of an exhaustive search decoder, algorithms such as the group LASSO also show a similar scaling of $mn=\Theta(k\log(d-k))$ provided $m>k$ \cite{Obozinski_annals_stats_2011}. From the discussion in the previous paragraph, it is clear that if we have $m=\Omega(k\log (d-k))$, then a single sample is sufficient for support recovery. 
Therefore, given that we have access to multiple samples now, a more interesting question to consider is whether we can perform recovery with $m<k$ measurements per sample. 
This measurement-constrained regime has received some attention in the past \cite{Balkan_SPL_2014}, \cite{Pal_TSP_2015}, \cite{Ramesh_ISIT_2019} and it was recently shown for the case of \emph{random} inputs drawn from a subgaussian distribution that the tradeoff (ignoring noise variance and parameters dependent on the generative model for the samples) is $n=\Theta((k^{2}/m^{2})\log d)$ for $(\log k)^{2}\le m<k/2$ \cite{Ramesh_arxiv_2019}.

\begin{figure}[t]
\centering 
\includegraphics[height=7cm,width=7cm]{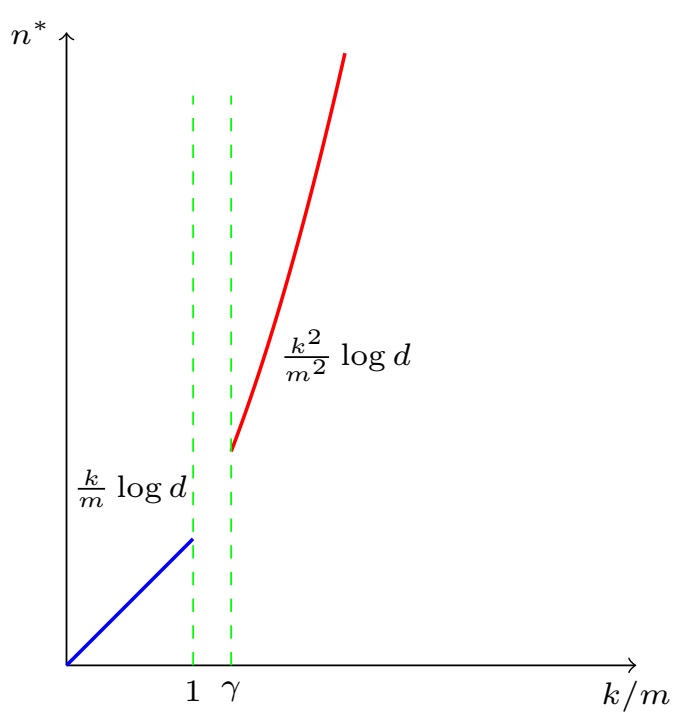}

\caption{Sample complexity of support recovery as a function of $k/m$.}
\end{figure}\label{fig:ptx}


In this work, we focus on the case of deterministic
inputs $x_{i}$ {with $|x_{ij}|\in[x_{\min},x_{\max}]$, $j\in\cS$,} and show that the tradeoff identified in~\cite{Ramesh_ISIT_2019}
for Gaussian inputs
holds for the worst-case setting as well. In particular, the lower bound from \cite{Ramesh_arxiv_2019} for Gaussian inputs applies to this worst-case setting, since an
instantiation in the Gaussian case can be thought of as a deterministic input.
Showing that the upper bound also remains the same requires more work, and is the main focus of this paper. 
Specifically, we analyze the performance of the estimator from \cite{Ramesh_arxiv_2019} in the deterministic input setting, which requires improved bounds on the tail probability of heavy-tailed random variables\footnote{We refer to a random variable $X$ as heavy-tailed if its moment generating function $\bEE{e^{\lambda(X-\bEE{X})}}$ is infinite for all $\lambda\in\mR$.} than the one used in \cite{Ramesh_arxiv_2019}.

In summary, we settle the question of tradeoff between $m$ and $n$ in the $m<k$ regime, and show that there exists a phase transition for the sample complexity of this problem at $k/m=1$ as depicted in Figure \ref{fig:ptx}. Roughly, around this point, the sample complexity for support recovery undergoes a change from being linear in the ratio $k/m$ to being quadratic in $k/m$ (up to a factor $\log d$).

We note that the current lower bound proof from \cite{Ramesh_arxiv_2019} requires some separation between $k$ and $m$; namely, it requires $k/m>\gamma$ for some $\gamma>1$.
While the lower bound of $n=\Omega((k/m)\log (d/k))$ \cite{Park_TIT_2017} continues to hold for $m<k$, it is not clear if a tighter lower bound on sample complexity in the regime $1< k/m\le \gamma$ can be obtained. Such a separation between $k$ and $m$ is, however, not required when deriving the upper bound. 


%
%
%
%
%
%
%
%
%
%
%
%
%



\textit{Notation.} We use upper case letters to denote either random variables or deterministic matrices, and lowercase letters to denote deterministic scalars or vectors. For a vector $x$, we use $x_{i}$ to denote its $i$th entry, and for matrices $\{A_{i}\}_{i=1}^{n}$ we use $A_{ij}$ to denote the $j$th column of $A_{i}$. We use $\|x\|_{p}^{p}\ed\sum_{i=1}^{d}x_{i}^{p}$ to denote the $\ell_{p}$ norm of a vector $x\in\mR^{d}$, and $\|Z\|_{\cL_{p}}^{p}\ed\bEE{|Z|^{p}}$ to denote the $\cL_{p}$ norm of a random variable $Z$.

\textit{Organization.} In the next section, we formally state the problem and present our main result. We present the proof of our main result in Section \ref{sec:upper_bound}, and end with a discussion on directions for further work in Section \ref{sec:discussion}.
\begin{onlycnfr}
Proofs that are omitted here can be found in \cite{LR-CM-HT-21a}.
\end{onlycnfr}
\section{Problem formulation and main result} \label{sec:formulation}
Let vectors $x_{1},\ldots,x_{n}$ in $\mR^{d}$ have a common support
$\cS\subset[d]$ of cardinality $k$. For each of these vectors, we have
access to noisy linear measurements of the form
$Y_{i}=\Phi_{i}x_{i}+W_{i}$, $i\in[n]$. Here, $\Phi_{i}\in\mR^{m\times
  d}$ with $m<d$ are called the measurement matrices and
$W_{i}\stackrel{iid}{\sim}\cN(0,\sigma^{2}I)$ is noise. 
The goal is to
recover the support $\cS$ using $\{Y_{i},\Phi_{i}\}_{i=1}^{n}$.  An estimator for $\cS$ is a mapping $\hat{\cS}:\mR^{m\times n}\times \mR^{m\times d\times n}\rightarrow \binom{[d]}{k}$, where $\binom{[d]}{k}$ denotes the set of all subsets of $[d]$ of cardinality $k$.
We assume that the estimator has knowledge of $k$ and 
consider the probability of exact recovery, $\bPr{\hat{\cS}\ne \cS}$, as our recovery criterion.
We note that one could also consider the setting where $|\cS|\le k$. The  estimator that we consider here would output an $\hat{\cS}$ that contains the true support with high probability.
In this work, however, we assume that the true support has cardinality exactly $k$.

We make the following two assumptions on the measurement matrices and the input samples:
\begin{assumption}\label{assump_phi}
	The {$m\times d$} measurement matrices
        $\Phi_{1},\ldots,\Phi_{n}$ are independent, with entries
        that are independent and distributed as $\cN(0,1/m)$.
\end{assumption}
\begin{assumption}\label{assump_x}
The $d$-dimensional inputs $x_{1},\ldots,x_{n}$ are such that $\mathtt{supp}(x_{i})=\cS$, for all $i\in[n]$, where $\cS\subset[d]$ is a fixed set of cardinality $k$. Further, 
$|x_{iu}|\in[x_{\min},x_{\max}]$, for all $i\in[n]$, $u\in\cS$, 
 where $x_{\min}, x_{\max}\in\mR$.
\end{assumption}

We focus on a measurement-constrained setting where we obtain only $m<k$ measurements per sample. The fundamental quantity of interest for us in this paper is the {\em sample complexity of support recovery,} defined below.
\begin{definition}
For $m,k,d\in\mN$, the sample complexity of support recovery $n^{*}(m,k,d,x_{\min},x_{\max},\sigma^{2},\delta)$ is the minimum number of samples $n$ for which we can find an estimator that can recover $\cS$ with probability of error at most $\delta$. Mathematically,
\begin{align}\label{eq:est_err}
\bPr{\hat{\cS}\ne\cS}\le\delta,\quad \forall \cS\in\binom{[d]}{k}.
\end{align}
\end{definition}

For notational convenience, we use $n^{*}(m,k,d,x_{\min},x_{\max},\sigma^{2},\delta)=n^{*}$ in the rest of the paper.
Our main result shows a phase transition that occurs at $k/m=1$ for the problem of support recovery. In particular, the dependence of sample complexity of support recovery on $k/m$
undergoes a sharp change from linear to quadratic
 as we move from the $k/m\le 1$ regime to the $k/m>1$ regime.
As mentioned before, a tight characterization of $n^{*}$ in the interval $1<k/m\le\gamma$ is not known, although our upper bound stated in Theorem \ref{thm} below continues to hold for this case. 
Our main result is the following.
\begin{theorem}\label{thm}
The sample complexity of support recovery under Assumptions \ref{assump_phi} and \ref{assump_x}, for $m\ge 2\log (d/\delta)$, satisfies
\begin{align*}
n^{*}
=O\bigg(\frac{x_{\max}^{4}}{x_{\min}^{4}}\max\bigg\{\bigg(\frac{k}{m}+\frac{\sigma^{2}}{x_{\max}^{2}}\bigg)\log \frac{d}{\delta},
\bigg(\frac{k}{m}+\frac{\sigma^{2}}{x_{\max}^{2}}\bigg)^{2}\log\frac{d}{\delta}\bigg\}\bigg).
\end{align*}
\end{theorem}
As a special case, in the noiseless setting with $m<k$, we have the following corollary, which follows from Theorem \ref{thm} above and the lower bound in \cite{Ramesh_arxiv_2019}\footnote{When $k\le d/2$ and $d\ge 4$, the $\log(k(d-k))$ factor in the lower bound is equal, upto constants, to $\log d$.}. Note that the lower bound is stated for a constant probability of error.
\begin{corollary}
In the noiseless setting, with $4\log 3d\le 2m<k\le d/2$, $d\ge 4$ and $\delta=1/3$, we have,
\begin{align*}
n^{*}
=\Theta\bigg(\frac{x_{\max}^{4}}{x_{\min}^{4}}\frac{k^{2}}{m^{2}}\log d\bigg).
\end{align*}
\end{corollary}

We provide the proof of Theorem \ref{thm} in the next section.

\section{Analysis of the estimator}\label{sec:upper_bound}
We will analyze the closed form estimator from \cite{Ramesh_arxiv_2019}, but instead of random inputs, here we will consider deterministic inputs $x_{1},\ldots,x_{n}$. To see why the analysis in \cite{Ramesh_arxiv_2019} does not extend in a straightforward way to this case, we first recall the form of the estimator. Let $\Phi_{iu}\in\mR^{m}$ denote the $u$th column of $\Phi_{i}$. We first compute proxy samples $\hat{X}_{1},\ldots,\hat{X}_{n}$ with entries
\begin{align}\label{eq:xhat}
\hat{X}_{iu}\ed\Phi_{iu}^{\top}Y_{i}=\Phi_{iu}^{\top}\Phi_{i}x_{i}+\Phi_{iu}^{\top}W_{i},\quad u\in[d],
\end{align} 
and then the sample second moment along each coordinate as 
\begin{align}\label{eq:lambda_est}
\tilde{\lambda}_{u}\ed\frac{1}{n}\sum_{i=1}^{n}\hat{X}_{iu}^{2},\quad u\in[d].
\end{align}
The support estimate $\tilde{\cS}$ consists of the $k$ indices of $\tilde{\lambda}$ with the largest value. 
Analyzing the estimator would basically involve obtaining tail bounds for the random variable above. Considering the noiseless case first, note that each summand in \eqref{eq:lambda_est} is of the form $(\Phi_{iu}^{\top}\Phi_{i}x_{i})^{2}$, and can be viewed as a quadratic in either $x_{i}$ or $\Phi_{iu}^{\top}\Phi_{i}$.

When $x_{i}$s are random and subgaussian with independent coordinates, we can exploit the quadratic form in $x_{i}$ to obtain a tail bound using the Hanson-Wright inequality (after conditioning on $\Phi_{i}$).
On the other hand, when $x_{i}$ are deterministic, the summands in \eqref{eq:lambda_est} are quadratic in $\Phi_{iu}^{\top}\Phi_{i}$,  resulting in a heavy-tailed random variable, and standard methods based on bounding the moment generating function (MGF) do not work.

We explain in the next section how a careful analysis involving conditioning on a certain \emph{column} of $\Phi_{i}$ followed by a moment based bound can be used to get exponential tail bounds for heavy-tailed random variables. The analysis in \cite{Ramesh_arxiv_2019} also deals with heavy-tailed random variables, but using a more elementary approach (see [Lemma B.2, \cite{Ramesh_arxiv_2019}] and [Lemma B.3, \cite{Ramesh_arxiv_2019}] for instance) which would not work here.

\subsection{A separation condition for support recovery}
We will analyze the error probability of a threshold-based version of the estimator described in the previous section. In particular, we will use the estimate $\hat{\lambda}\ed\mathbbm{1}_{\{\tilde{\lambda}\ge\tau\}}$, for an appropriate threshold $\tau$, since 
$\bPr{\tilde{\cS}\ne \cS}\le \bPr{\hat{\cS}\ne \cS}$,
where $\hat{\cS}$ denotes the support of $\hat{\lambda}$. 
The error probability $\bPr{\hat{\cS}\ne \cS}$ will essentially be determined by the tail behaviour of the variance estimate $\tilde{\lambda}$. Recall from the last section that variance estimate is an average of random variables of the form $(\Phi_{iu}^{\top}\Phi_{i}x_{i}+\Phi_{iu}^{\top}W_{i})^{2}$. The $\Phi_{iu}^{\top}\Phi_{i}x_{i}$ term will be indicative of whether the coordinate $u$ lies in the support or not, since it will have a $\|\Phi_{iu}\|_{2}^{2}$
 term only when $u\in\cS$. 
 
The analysis is greatly simplified once we condition on $\Phi_{iu}$, because then the summands in \eqref{eq:lambda_est} are noncentral chi-square distributed, for which tail bounds can be obtained using standard methods. 
The error probability can be made small provided these tail probabilities (parameterized by $\Phi_{iu}$) can be made small, 
which eventually leads to a condition on the measurement ensemble. We will show, using tail bounds for heavy-tailed random variables, that this condition is satisfied with high probability for the Gaussian ensemble when the parameters $(n,m,k,d)$ scale as indicated in Theorem \ref{thm}, thus finishing the proof.

The probability of error can be bounded as
\begin{align}\label{eq:err}
\bPr{\hat{\cS}\ne\cS}
&\le \sum_{u\in\cS}\bPr{\tilde{\lambda}_{u}<\tau|E}
+\sum_{u^{\prime}\in\cS^{c}}\bPr{\tilde{\lambda}_{u^{\prime}}\ge\tau|E}+\bPr{E^{c}},
\end{align}
where $E$ denotes the event that the measurement ensemble satisfies a certain condition, which we will describe shortly.
For the right hand side to remain below $\delta$, we require the summands in the first two terms to be at most $\delta/(3\max\{k,d-k\})$. For simplicity, we will work with a requirement of $\delta/3d$.
 Now, using \eqref{eq:xhat} and \eqref{eq:lambda_est}, we can see that $\hat{X}_{iu}|\Phi_{iu}\sim\cN(\mu_{i},\nu_{i}^{2})$ for $u\in\cS$ with 
\begin{align*}
\mu_{i}=\|\Phi_{iu}\|_{2}^{2}x_{iu},
\end{align*}
and
\begin{align*}
\nu_{i}^{2}=\frac{\|\Phi_{iu}\|_{2}^{2}}{m}\sum_{v\in\cS\backslash\{u\}}x_{iv}^{2}+\sigma^{2}\|\Phi_{iu}\|_{2}^{2}.
\end{align*}
Similarly, we have $\hat{X}_{iu^{\prime}}|\Phi_{iu^{\prime}}\sim\cN(0,\nu_{i}^{\prime 2})$, for $u^{\prime}\in\cS^{c}$, where
\begin{align*}
\nu_{i}^{\prime 2}=\frac{\|\Phi_{iu^{\prime}}\|_{2}^{2}}{m}\sum_{v\in\cS}x_{iv}^{2}+\sigma^{2}\|\Phi_{iu^{\prime}}\|_{2}^{2}.
\end{align*}
A direct application of Lemma \ref{lem:chi_squared} then yields, for every $u\in\cS$,
\begin{align*}
\bPr{\tilde{\lambda}_{u}<\tau|\{\Phi_{iu}\}_{i=1}^{n}}
\le\exp\bigg(\frac{-n^{2}(\mu-\tau)^{2}}{4(\sum_{i=1}^{n}\nu_{i}^{4}+\nu_{i}^{2}\mu_{i}^{2})}\bigg),
\end{align*}
where $\mu\ed\bEE{\tilde{\lambda}_{u}|\{\Phi_{iu}\}_{i=1}^{n}}$. For $u^{\prime}\in\cS^{c}$, we can obtain in a similar manner from Lemma \ref{lem:chi_squared},
\begin{align*}
&\bPr{\tilde{\lambda}_{u^{\prime}}\ge\tau|\{\Phi_{iu^{\prime}}\}_{i=1}^{n}}
\le \exp\bigg(-\min\bigg\{\frac{n^{2}(\tau-\mu^{\prime})^{2}}{16\sum_{i=1}^{n}\nu_{i}^{\prime 4}},\frac{n(\tau-\mu^{\prime})}{8\max_{i\in[n]}\nu_{i}^{\prime 2}}\bigg\}\bigg),
\end{align*}
where $\mu^{\prime}\ed\bEE{\tilde{\lambda}_{u^{\prime}}|\{\Phi_{iu}\}_{i=1}^{n}}$. For the missed detection and false alarm probabilities above to remain bounded above by $\delta/3d$, we require
\begin{align*}
\tau \le \mu -\sqrt{\frac{4}{n^{2}}\sum_{i=1}^{n}(\nu_{i}^{4}+\mu_{i}^{2}\nu_{i}^{2})\log\frac{3d}{\delta}},
\end{align*}
and
\begin{align*}
\tau \ge \mu^{\prime}+\max\bigg\{\sqrt{\frac{16}{n^{2}}\sum_{i=1}^{n}\nu_{i}^{\prime 4}\log\frac{3d}{\delta}},\frac{8}{n}\max_{i\in[n]}\nu_{i}^{\prime 2}\log\frac{3d}{\delta}\bigg\}.
\end{align*}
Therefore, for the existence of a threshold $\tau$, we can see upon simplification that it suffices to have
\begin{align}\label{eq:sep_0}
\mu-\mu^{\prime}
> \sqrt{\frac{4}{n^{2}}\sum_{i=1}^{n}(\nu_{i}^{4}
+\nu_{i}^{2}\mu_{i}^{2})\log\frac{3d}{\delta}}
+\max\bigg\{\sqrt{\frac{16}{n^{2}}\sum_{i=1}^{n}\nu_{i}^{\prime 4}\log\frac{3d}{\delta}},\frac{8}{n}\max_{i\in[n]}\nu_{i}^{\prime 2}\log\frac{3d}{\delta}\bigg\}.
\end{align}
A simple calculation shows that the conditional mean of the estimator under the $u\in\cS$ and $u^{\prime}\in\cS^{c}$ cases are separated roughly by a constant term (after averaging over the measurement matrices), which makes the distinction between the two cases possible. In particular,
 \begin{align*}
 \mu=\frac{1}{n}\sum_{i=1}^{n}\bigg(x_{iu}^{2}\|\Phi_{iu}\|_{2}^{4} 
+\|\Phi_{iu}\|_{2}^{2} \bigg(\frac{1}{m}\sum_{v\in\cS\backslash\{u\}}x_{iv}^{2}+\sigma^{2}\bigg)\bigg),
 \end{align*}
  and
  \begin{align*}
   \mu^{\prime}=\frac{1}{n}\sum_{i=1}^{n}\|\Phi_{iu}\|_{2}^{2}\bigg(\frac{1}{m}\sum_{v\in\cS}x_{iv}^{2}+\sigma^{2}\bigg).
  \end{align*} 
Substituting this into \eqref{eq:sep_0} and simplifying, we can rewrite the condition as
\begin{align}\label{eq:sep}
&\frac{x_{\min}^{2}}{x_{\max}^{2}}\frac{1}{n}\sum_{i=1}^{n}\bigg(\|\Phi_{iu}\|_{2}^{4}-\frac{1}{m}\|\Phi_{iu}\|_{2}^{2}\bigg)
>
\sqrt{\frac{4}{n^{2}}\bigg(\frac{k-1}{m}+\frac{\sigma^{2}}{x_{\max}^{2}}\bigg)^{2}\sum_{i=1}^{n}\|\Phi_{iu}\|_{2}^{4}\log\frac{3d}{\delta}}\nonumber\\
&+\sqrt{\frac{4}{n^{2}}\bigg(\frac{k-1}{m}+\frac{\sigma^{2}}{x_{\max}^{2}}\bigg)\sum_{i=1}^{n}\|\Phi_{iu}\|_{2}^{6}\log\frac{3d}{\delta}}
+\sqrt{\frac{16}{n^{2}}\bigg(\frac{k}{m}+\frac{\sigma^{2}}{x_{\max}^{2}}\bigg)^{2}\sum_{i=1}^{n}\|\Phi_{iu^{\prime}}\|_{2}^{4}\log\frac{3d}{\delta}}\nonumber\\
&+\frac{8}{n}\bigg(\frac{k}{m}+\frac{\sigma^{2}}{x_{\max}^{2}}\bigg)\max_{i\in[n]}\|\Phi_{iu^{\prime}}\|_{2}^{2}\log\frac{3d}{\delta},
\end{align}
for every $(u,u^{\prime})\in\cS\times \cS^{c}$.

\subsection{Separation condition for the Gaussian ensemble}
We will show that when the measurement ensemble is Gaussian as described in Assumption \ref{assump_phi}, the separation condition in \eqref{eq:sep} is satisfied with high probability for a certain regime of the parameters $(n,m,k,d)$. 
We will derive upper and lower bounds on the right hand side and left hand side respectively in \eqref{eq:sep}, that hold with high probability, which after simplification will finally result in a condition on the parameters as stated in Theorem \ref{thm}. Note that this translates to obtaining tail bounds for the random variable $ (1/n)\sum_{i=1}^{n}\|\Phi_{iu}\|_{2}^{2q}$ with $q=2,3$. 
It is easy to see that $\|\Phi_{iu}\|_{2}^{2}$ is chi-square distributed (after scaling by $m$), and $\|\Phi\|_{2}^{2q}$ is therefore a heavy-tailed random variable, and so MGF based methods cannot be used here. We will see that a bound on the moments can be used to get exponential tail bounds (even when the MGF is unbounded).
\begin{onlyarxiv}
The proofs for results in this section can be found in Appendix \ref{app-proofs}.
\end{onlyarxiv}

We will fix $q=3$ and derive our results; the same arguments can be used for the $q=2$ case as well. Define $Z\ed |(1/n)\sum_{i=1}^{n}(\|\Phi_{iu}\|_{2}^{6}-\bEE{\|\Phi_{iu}\|_{2}^{6}}|$ and note that for all $p\ge 1$,
\begin{align}\label{eq:tail_bound_moment}
\bPr{Z\ge e(\bEE{Z^{p}})^{\frac{1}{p}}}
=\bPr{Z^{p}\ge e^{p}\bEE{Z^{p}}}\le e^{-p}. 
\end{align}
Further, for all $p\ge 2$, if we can show that $(\bEE{Z^{p}})^{\frac{1}{p}}\le cp^{\beta}$ for some $\beta>0$, then together with the previous inequality it implies that $\bPr{Z\ge ecp^{\beta}}\le e^{-p}$, or, equivalently, for $t>0$, that
\begin{align}\label{eq:tail_bound_moment_2}
\bPr{Z\ge t}\le \exp(-(t/ec)^{\frac{1}{\beta}}).
\end{align}
We now need to determine an upper bound on $\|Z\|_{\cL_{p}}\ed (\bEE{Z^{p}})^{\frac{1}{p}}$. We show such a moment bound, resulting in the following lemma. 
\begin{onlycnfr}
The proof can be found in \cite{LR-CM-HT-21a}.
\end{onlycnfr}

\begin{lemma}\label{lem:heavy_tailed}
For every $t>0$, there exists an absolute constant $C$ such that
\begin{align*}
\bPr{\bigg|\frac{1}{n}\sum_{i=1}^{n}(\|\Phi_{iu}\|_{2}^{6}-\bEE{\|\Phi_{iu}\|_{2}^{6}})\bigg|\ge t}
\le \exp\bigg(-C\min\bigg\{nt,(m^{3}nt)^{\frac{1}{4}},nt^{2}\bigg\}\bigg).
\end{align*}
\end{lemma}
\begin{onlycnfr}
{In \cite{LR-CM-HT-21a}, we provide a simple, self-contained proof of this lemma.
We note that since $\|\Phi_{iu}\|_{2}^{6}$ is a polynomial in $m$ i.i.d. Gaussian random variables, results such as \cite[Theorem 1.3]{Adamczak_PTRF_2013} can be used to obtain tail bounds for some of the terms. The proof in \cite{LR-CM-HT-21a} is more straightforward.}
\end{onlycnfr}

\begin{onlyarxiv}
We note that since $\|\Phi_{iu}\|_{2}^{6}$ is a polynomial in $m$ i.i.d. Gaussian random variables, results such as \cite[Theorem 1.3]{Adamczak_PTRF_2013} can be used to obtain tail bounds for some of the terms. Our proof in Appendix \ref{app-proofs} is more straightforward.
\end{onlyarxiv}

\begin{onlyarxiv}
A similar result can be obtained for the $(1/n)\sum_{i=1}^{n}\|\Phi_{iu}\|_{2}^{4}$ term in \eqref{eq:sep} using the same technique, and we omit the proof for this result.
\begin{lemma}\label{lem:heavy_tailed_2}
For every $t>0$, there exists an absolute constant $C$ such that
\begin{align*}
\bPr{\bigg|\frac{1}{n}\sum_{i=1}^{n}(\|\Phi_{iu}\|_{2}^{4}-\bEE{\|\Phi_{iu}\|_{2}^{4}})\bigg|\ge t}
\le \exp\bigg(-C\min\bigg\{nt,(m^{2}nt)^{\frac{1}{3}},nt^{2}\bigg\}\bigg).
\end{align*}
\end{lemma}
\end{onlyarxiv}

Together with the fact that $\bEE{\|\Phi_{iu}\|_{2}^{4}}=1+2/m$ and $\bEE{\|\Phi_{iu}\|_{2}^{6}}=1+6/m+8/m^{2}$, the results above
give upper and lower bounds that hold with high probability on all but the $\max_{i\in[n]}\|\Phi_{iu^{\prime}}\|_{2}^{2}$  term in \eqref{eq:sep}. The latter can be bounded with high probability using concentration for chi-squared random variables and a union bounding step, as given by the following lemma.
\begin{onlycnfr}
The proof can be found in \cite{LR-CM-HT-21a}.
\end{onlycnfr}

\begin{lemma}\label{lem:max_chi_squared}
Let $\mu_{\max}\ed\bEE{\max_{i\in[n]}\|\Phi_{iu}\|_{2}^{2}}$. Then, for every $t>0$,
\begin{align*}
\bPr{\max_{i\in[n]}\|\Phi_{iu}\|_{2}^{2}\ge \mu_{\max}+t}
\le n\exp\bigg(\frac{-m}{8}\min\bigg\{(\mu_{\max}+t-1)^{2},\mu_{\max}+t-1\bigg\}\bigg).
\end{align*}
\end{lemma}

To ensure that the random variable on the left hand side of \eqref{eq:sep} exceeds the one on the right hand side with large probability, we can substitute the bounds we derived for each term, and check when the inequality holds. This results (up to some constant loss in the $\delta$ factor) in a condition on the problem parameters under which \eqref{eq:sep} holds for a fixed $(u,u^{\prime})\in\cS\times \cS^{c}$. Applying a union bound over all $k(d-k)$ pairs gives the final requirement on $n$. 
\begin{onlycnfr}
Note that the leading terms on the right hand side of \eqref{eq:sep} would roughly be $\sqrt{(k^{2}/m^{2}n)\log d/\delta}$ or $\sqrt{(k/mn)\log d/\delta}$ (assuming $m\ge 2 \log (d/\delta)$, see \cite{LR-CM-HT-21a} for details of the proof), while the left hand side would roughly be a constant, leading to the following result.
\end{onlycnfr}
\begin{onlyarxiv}
Note that the leading terms on the right hand side of \eqref{eq:sep} would roughly be $\sqrt{(k^{2}/m^{2}n)\log d/\delta}$ or $\sqrt{(k/mn)\log d/\delta}$ (assuming $m\ge 2 \log (d/\delta)$, see the proof in Appendix \ref{app-proofs} for details), while the left hand side would roughly be a constant, leading to the following result.
\end{onlyarxiv}

\begin{lemma}\label{lem:sep}
The separation condition \eqref{eq:sep} holds for every $(u,u^{\prime})\in\cS\times \cS^{c}$, with probability at least $1-\delta$, provided $m\ge 2\log(d/\delta)$ and
\begin{align*}
n\ge c\frac{x_{max}^{4}}{x_{\min}^{4}}\max\bigg\{\bigg(\frac{k}{m}+\frac{\sigma^{2}}{x_{\max}^{2}}\bigg)\log\frac{d}{\delta},
\bigg(\frac{k}{m}+\frac{\sigma^{2}}{x_{\max}^{2}}\bigg)^{2}\log\frac{d}{\delta}\bigg\},
\end{align*}
for an absolute constant $c$.
\end{lemma}

By defining $E$ as the event that the measurement matrices satisfy condition \eqref{eq:sep} for every $(u,u^{\prime})\in\cS\times \cS^{c}$, we can see that the probability of error in \eqref{eq:err} is at most $\delta$, provided $n$ satisfies the condition in Lemma \ref{lem:sep}. This completes the proof of Theorem \ref{thm}.

\section{Discussion}\label{sec:discussion}
We showed a phase transition for the problem of support recovery from multiple samples.
While the closed form estimator that we analyzed here is sample-optimal, it would be interesting to design other estimators that can work in the measurement-constrained regime without knowledge of the support size, and for which guarantees can be obtained with worst-case inputs.  
Finally, extending the lower bound on $n^{*}$ to include the $1<k/m\le \gamma$ regime for $\gamma>1$ would provide a better understanding of the problem.

\begin{appendices}
\centering
\section{}

\begin{lemma-app}\label{lem:chi_squared}
Let $X_{1}, \ldots,X_{n}$ be drawn i.i.d. from $\cN(\mu_{i},\sigma_{i}^{2})$. Then, for every $t>0$,
\begin{align*}
\bPr{\frac{1}{n}\sum_{i=1}^{n}X_{i}^{2}
\le \frac{1}{n}\sum_{i=1}^{n}(\sigma_{i}^{2}+\mu_{i}^{2})- t}
\le
\exp\bigg(\frac{-n^{2}t^{2}}{4\sum_{i=1}^{n}(\sigma_{i}^{4}+\sigma_{i}^{2}\mu_{i}^{2})}\bigg),
\end{align*}
and
\begin{align*}
\bPr{\frac{1}{n}\sum_{i=1}^{n}X_{i}^{2}
\ge \frac{1}{n}\sum_{i=1}^{n}(\sigma_{i}^{2}+\mu_{i}^{2})+ t}
\le \exp\bigg(-\min\bigg\{\frac{n^{2}t^{2}}{16\sum_{i=1}^{n}(\sigma_{i}^{4}+\sigma_{i}^{2}\mu_{i}^{2})},\frac{nt}{8\underset{i \in [n]}{\max}~\sigma_{i}^{2}}\bigg\}\bigg).
\end{align*}
\end{lemma-app} 
\begin{onlyarxiv}
\begin{proof}
The proof is similar to that of \cite{Birge_AoS_2001} for $\sigma^{2}=1$, and follows by upper bounding the MGF of a noncentral chi-squared random variable and then using the Chernoff method. We include the proof here for completeness.
We will first show the left tail bound. To that end, we note that for $t>0$ and $\lambda<0$, the following holds for $Y\ed (1/n)\sum_{i=1}^{n}X_{i}^{2}$:
\begin{align*}\label{eq:chernoff}
\bPr{Y\le \bEE{Y}-t}
\le e^{\lambda t}\bEE{e^{\lambda(Y-\bEE{Y})}}.\numberthis
\end{align*}
To upper bound the MGF, first note that for $X\sim\cN(\mu,\sigma^{2})$,
\begin{align*}
\bEE{e^{\lambda (X^{2}-\bEE{X^{2}})}}
&=e^{-\lambda(\sigma^{2}+\mu^{2})}\frac{1}{\sqrt{2\pi}\sigma}\int_{-\infty}^{\infty}e^{\lambda x^{2}}e^{\frac{-(x-\mu)^{2}}{2\sigma^{2}}}dx\\
&=\frac{e^{-\lambda(\sigma^{2}+\mu^{2})}}{\sqrt{1-2\lambda\sigma^{2}}}e^{\frac{\lambda \mu^{2}}{1-2\lambda\sigma^{2}}},
\end{align*}
 for all $\lambda<1/2\sigma^{2}$. Taking logarithms we have
 \begin{align*}\label{eq:log_mgf}
 \log \bEE{e^{\lambda (X^{2}-\bEE{X^{2}})}}
 &=\frac{1}{2}\bigg(-\log (1-2\lambda\sigma^{2}-2\lambda\sigma^{2})\bigg)+\frac{2\lambda^{2}\mu^{2}\sigma^{2}}{1-2\lambda\sigma^{2}}\numberthis \\
 &\le \lambda^{2}\sigma^{4}+\frac{2\lambda^{2}\mu^{2}\sigma^{2}}{1-2\lambda\sigma^{2}}\\
 &\le \lambda^{2}(\sigma^{4}+2\mu^{2}\sigma^{2}),
 \end{align*}
 where we used $-\log(1-x)-x\le x^{2}/2$ for $x<0$ in the second step. This gives
 \begin{align*}
 \log \bEE{e^{\lambda (Y-\bEE{Y})}}
 &\le \frac{\lambda^{2}}{n^{2}}\sum_{i=1}^{n}(\sigma_{i}^{4}+\sigma_{i}^{2}\mu_{i}^{2})
 \end{align*}
 which upon substituting into \eqref{eq:chernoff} and optimizing over $\lambda<0$ gives $\lambda=-n^{2}t/(2\sum_{i=1}^{n}(\sigma_{i}^{4}+\sigma_{i}^{2}\mu_{i}^{2})))$ resulting in the left tail bound claimed in the lemma.
 
 For the right tail bound, we continue from \eqref{eq:log_mgf} and note that for $0\le \lambda\le 1/4\sigma^{2}$,
 \begin{align*}
 \log \bEE{e^{\lambda(X^{2}-\bEE{X^{2}})}}
 &\le \frac{2\lambda^{2}\sigma^{4}}{1-2\lambda\sigma^{2}}+\frac{2\lambda^{2}\mu^{2}\sigma^{2}}{1-2\lambda\sigma^{2}}\\
 &\le 4\lambda^{2}(\sigma^{4}+\mu^{2}\sigma^{2}),
 \end{align*} 
 where in the first step we used $-\log(1-x)-x\le x^{2}/2(1-x)$ for all $x\in[0,1)$. Extending as before to the normalized sum $(1/n)\sum_{i=1}^{n}X_{i}^{2}$, substituting into \eqref{eq:chernoff} and optimizing over $\lambda\in[0,1/4\sigma^{2})$, it can be seen that the minimum is attained at $\lambda=nt/(8\sum_{i=1}^{n}(\sigma_{i}^{4}+\mu_{i}^{2}\sigma^{2}))$  if $t<2\sum_{i=1}^{n}(\sigma^{2}+\mu^{2})$, and at $\lambda=1/(4\sum_{i=1}^{n}\sigma_{i}^{2})$ otherwise. This gives the right tail bound claimed in the lemma.
\end{proof}
\end{onlyarxiv}

\begin{onlyarxiv}
\section{}\label{app-proofs}

\begin{proof}[Proof of Lemma \ref{lem:heavy_tailed}]
First, note that
\begin{align}\label{eq:moment_bound_pf}
\bigg\|\frac{1}{n}\sum_{i=1}^{n}(\|\Phi_{iu}\|_{2}^{6}-\bEE{\|\Phi_{iu}\|_{2}^{6}})\bigg\|_{\cL_{p}}
=\frac{1}{nm^{3}}\bigg\|\sum_{i=1}^{n}(V_{i}^{3}-\bEE{V_{i}^{3}})\bigg\|_{\cL_{p}},
\end{align}
where $V_{i}\ed m\|\Phi_{iu}\|_{2}^{2}\sim\chi_{m}^{2}$, and $\chi_{m}^{2}$ denotes the chi-square distribution with $m$ degrees of freedom. 
To bound the moment of the sum, we use the following form of Rosenthal's inequality stated in \cite{Pinelis_TPA_1985}.

\begin{lemma-app}[{\hspace{1sp}\cite{Pinelis_TPA_1985}}]\label{lem:Rosenthal}
Let $Z_{1},\ldots,Z_{n}$ be independent and identically distributed random variables with mean zero. Then, for every $p\ge 2$,
\begin{align*}
\bigg\|\sum_{i=1}^{n}Z_{i}\bigg\|_{\cL_{p}}
\le c\bigg(pn^{\frac{1}{p}}\|Z_{1}\|_{\cL_{p}}+\sqrt{pn}\|Z_{1}\|_{\cL_{2}}\bigg),
\end{align*}
for an absolute constant $c$.
\end{lemma-app}
In view of Lemma \ref{lem:Rosenthal}, we now upper bound the $\cL_{p}$ norm of each summand on the right side of \eqref{eq:moment_bound_pf} as follows:
\begin{align*}
\|V_{i}^{3}-\bEE{V_{i}^{3}}\|_{\cL_{p}}
&\le \|V_{i}^{3}\|_{\cL_{p}}+\bEE{V_{i}^{3}}\\
&=(\bEE{V_{i}^{3p}})^{\frac{1}{p}}+\bEE{V_{i}^{3}}\\
&=\bigg(2^{3p}\frac{\Gamma(3p+m/2)}{\Gamma(m/2)}\bigg)^{\frac{1}{p}}+2^{3}\frac{\Gamma(3+m/2)}{\Gamma(m/2)}\\
&\le 2^{3}\bigg(e^{\frac{1}{p}}(3p+m/2)^{3}+e(3+m/2)^{3}\bigg)\\
&\le 2^{6}(3p+m/2)^{3},
\end{align*}
where we used the fact that $V_{i}\sim\chi_{m}^{2}$ in the third step and $\Gamma(x+a)/\Gamma(x)\le e(x+a)^{a}$ for all $x\ge 1$, $a>0$ in the fourth step.
Together with Lemma \ref{lem:Rosenthal},  
this yields for $p\ge 2$,
\begin{align*}\label{eq:moment_bound}
\bigg\|\frac{1}{n}\sum_{i=1}^{n}(\|\Phi_{iu}\|_{2}^{6}-\bEE{\|\Phi_{iu}\|_{2}^{6}})\bigg\|_{\cL_{p}}
&\le \frac{c2^{6}}{nm^{3}}\bigg(pn^{\frac{1}{p}}(3p+m/2)^{3}+\sqrt{pn}(6+m/2)^{3}\bigg)\\
&\le c2^{6}\bigg(\frac{p}{n^{1-\frac{1}{p}}}\max\bigg\{1,\frac{(6p)^{3}}{m^{3}}\bigg\}+7^{3}\sqrt{\frac{p}{n}}\bigg)\\
&\le c^{\prime}\max\bigg\{\frac{p}{n^{1-\frac{1}{p}}},\frac{p^{4}}{m^{3}n^{1-\frac{1}{p}}},\sqrt{\frac{p}{n}}\bigg\}\numberthis.
\end{align*}

Note that from \eqref{eq:tail_bound_moment}, we expect $p$ to be of the form $n^{c^{\prime\prime}}$ for some constant $c^{\prime\prime}$, in which case $p/n^{1-\frac{1}{p}}=(p/n)e^{\frac{1}{ec^{\prime\prime}}}$. We focus on this regime, and obtain using \eqref{eq:tail_bound_moment_2} and \eqref{eq:moment_bound},
\begin{align*}
\bPr{\bigg|\frac{1}{n}\sum_{i=1}^{n}(\|\Phi_{iu}\|_{2}^{6}-\bEE{\|\Phi_{iu}\|_{2}^{6}})\ge t\bigg|}
\le \exp\bigg(-C\min\bigg\{nt,(m^{3}nt)^{\frac{1}{4}},nt^{2}\bigg\}\bigg),
\end{align*}
for every $t>0$.
\end{proof}

\end{onlyarxiv}
\begin{onlyarxiv}
\begin{proof}[Proof of Lemma \ref{lem:max_chi_squared}]
Let $\mu_{\max}=\max_{i\in[n]}\|\Phi_{iu}\|_{2}^{2}$. The proof follows by noting that for every $t>0$,
\begin{align*}
\bPr{\max_{i\in[n]}\|\Phi_{iu}\|_{2}^{2}\ge\mu_{\max}+ t}\le \sum_{i=1}^{n}\bPr{\|\Phi_{iu}\|_{2}^{2}-1\ge t^{\prime}},
\end{align*}
where $t^{\prime}=\mu_{\max}+t-1$, and using the fact that $m\|\Phi_{iu}\|_{2}^{2}\sim\chi_{m}^{2}$ to get
\begin{align*}
\bPr{\max_{i\in[n]}\|\Phi_{iu}\|_{2}^{2}\ge t}
\le\exp\bigg(-\frac{m}{8}\min\{t^{\prime 2},t^{\prime}\}\bigg).
\end{align*}
\end{proof}
\end{onlyarxiv}
\begin{onlyarxiv}
\begin{proof}[Proof of Lemma \ref{lem:sep}]
The proof involves finding upper and lower bounds, respectively, on the left-hand side and right-hand side of \eqref{eq:sep} that hold with high probability, and then simplifying to obtain the condition on $n$ stated in the lemma. 
Note that there are two probability of error parameters here, one from the criterion in \eqref{eq:err}, and another required for \eqref{eq:sep}. To avoid confusion, will use $\delta$ for the former and $\delta^{\prime}$ for the latter (we will eventually set $\delta^{\prime}=\delta/(k(d-k))$).

For the left-hand side of \eqref{eq:sep}, it follows from Lemma \ref{lem:heavy_tailed_2} that
\begin{align*}
\bPr{\frac{1}{n}\sum_{i=1}^{n}\|\Phi_{iu}\|_{2}^{4}\ge 1+\frac{2}{m}-t}\le \delta^{\prime},
\end{align*}
when
\begin{align*}
t\ge \frac{1}{C}\max\bigg\{\frac{1}{n}\log\frac{1}{\delta^{\prime}},\frac{1}{nm^{2}}\bigg(\log\frac{1}{\delta^{\prime}}\bigg)^{3},\sqrt{\frac{1}{n}\log\frac{1}{\delta^{\prime}}}\bigg\},
\end{align*}
where the maximum in th expression above is the third term provided $m>\log(1/\delta^{\prime})$.
Further, since $m\|\Phi_{iu}\|_{2}^{2}\sim\chi_{m}^{2}$, we have
\begin{align*}
\bPr{\frac{1}{mn}\sum_{i=1}^{n}\|\Phi_{iu}\|_{2}^{2}\ge \frac{1}{m}-t}\le \delta^{\prime},
\end{align*}
when
\begin{align*}
t\ge \frac{2}{m}\bigg(\sqrt{\frac{1}{mn}\log\frac{1}{\delta^{\prime}}}+\frac{1}{mn}\log\frac{1}{\delta^{\prime}}\bigg).
\end{align*}
It follows that the left-hand side of \eqref{eq:sep} is at least $cx_{\min}^{2}/x_{\max}^{2}$ with probability at least $2\delta^{\prime}$, for an absolute constant $c$, provided $m>\log(1/\delta^{\prime})$.


We now proceed to find a high probability upper bound on the right-hand side of \eqref{eq:sep}.
Lemmas \ref{lem:heavy_tailed} and \ref{lem:heavy_tailed_2} can be used to upper bound the first three terms, and Lemma \ref{lem:max_chi_squared} can be used for the last term. In particular, we have
\begin{align*}
\bPr{\frac{1}{n}\sum_{i=1}^{n}\|\Phi_{iu}\|_{2}^{6}\ge 1+\frac{6}{m}+\frac{8}{m^{2}}+t}\le \delta^{\prime},
\end{align*}
when
\begin{align*}
t\ge \frac{1}{C}\max\bigg\{\frac{1}{n}\log\frac{1}{\delta^{\prime}},\frac{1}{nm^{3}}\bigg(\log\frac{1}{\delta^{\prime}}\bigg)^{4},\sqrt{\frac{1}{n}\log\frac{1}{\delta^{\prime}}}\bigg\}.
\end{align*}
Further,
\begin{align*}
\bPr{\frac{1}{n}\sum_{i=1}^{n}\|\Phi_{iu}\|_{2}^{4}\ge 1+\frac{2}{m}+t}\le \delta^{\prime},
\end{align*}
when
\begin{align*}
t\ge \frac{1}{C}\max\bigg\{\frac{1}{n}\log\frac{1}{\delta^{\prime}},\frac{1}{nm^{2}}\bigg(\log\frac{1}{\delta^{\prime}}\bigg)^{3},\sqrt{\frac{1}{n}\log\frac{1}{\delta^{\prime}}}\bigg\},
\end{align*}
and
\begin{align*}
&\bPr{\max_{i\in[n]}\|\Phi_{iu}\|_{2}^{2}\ge 1+\max\bigg\{\sqrt{\frac{8}{m}\log\frac{n}{\delta^{\prime}}},\frac{8}{m}\log\frac{n}{\delta^{\prime}}\bigg\}}\le\delta^{\prime}. 
\end{align*}
To simplify the right-hand side of \eqref{eq:sep}, note that after substituting the bounds above, the leading terms arise from the mean of $\Phi_{i}$ dependent terms (i.e. normalized sum and the normalized maximum), which is roughly $1$. 
In particular, we see that the leading terms are roughly $k/mn\cdot\log(1/\delta)$ and $k^{2}/m^{2}\cdot\log(d/\delta)$, provided $m\ge \log(1/\delta^{\prime})$ (this condition ensures that the deviation terms for the $\Phi_{i}$ dependent terms are small). 
Using this observation and recalling that the left-hand side in \eqref{eq:sep} is a constant gives, after simplification, that \eqref{eq:sep} holds for a fixed $(u,u^{\prime})\in\cS\times\cS^{c}$ with probability at least $1-\delta^{\prime}$, provided $m\ge\log(1/\delta^{\prime})$ and
\begin{align*}
n\ge c\frac{x_{max}^{4}}{x_{\min}^{4}}\max\bigg\{\bigg(\frac{k}{m}+\frac{\sigma^{2}}{x_{\max}^{2}}\bigg)\log\frac{d}{\delta},\bigg(\frac{k}{m}+\frac{\sigma^{2}}{x_{\max}^{2}}\bigg)^{2}\log\frac{d}{\delta}\bigg\},
\end{align*}
for an absolute constant $c$. We now apply a union bound over all pairs $(u,u^{\prime})$ and choose $\delta^{\prime}=\delta/(k(d-k))$. Finally, noting that $\log(1/\delta^{\prime})\le 2\log(d/\delta)$ , gives us the result stated in the lemma.
\end{proof}
\end{onlyarxiv}


\end{appendices}


\bibliography{IEEEabrv,bibfile,bibJournalList}
\bibliographystyle{IEEEtranS}

\end{document}